\documentclass[conference,twocolumn]{IEEEtran}
\usepackage[T1]{fontenc}% optional T1 font encoding
\usepackage{datetime}
\usepackage{listings}
\usepackage{graphicx}
\usepackage{float}
\usepackage{subfigure}
\usepackage{cite}
\usepackage{caption}
\usepackage{subfigure}
\usepackage{color}
\usepackage{amsthm}

\newtheorem{prop}{Proposition}

\usepackage[includeall, textwidth = 520 pt, textheight=672 pt, columnsep = 0.228 in]{geometry}

\usepackage{amsmath} 
   
\usepackage[cmintegrals]{newtxmath} 

\begin{document}
\title{On the Age of Information for Multicast Transmission with Hard Deadlines in IoT Systems}
\author{\IEEEauthorblockN{Jie Li\IEEEauthorrefmark{1}\IEEEauthorrefmark{2}\IEEEauthorrefmark{3}, Yong Zhou\IEEEauthorrefmark{1}, and He Chen\IEEEauthorrefmark{4}}

\IEEEauthorblockA{\IEEEauthorrefmark{1}School of Information Science and Technology, ShanghaiTech University, Shanghai 201210, China \hspace{3mm}\\
\IEEEauthorrefmark{2}Shanghai Institute of Microsystem and Information Technology, Chinese Academy of Sciences, China \\
\IEEEauthorrefmark{3}University of Chinese Academy of Sciences, Beijing, China \;\;
\IEEEauthorrefmark{4}The Chinese University of Hong Kong, Hong Kong, China \\
Email: {\{lijie3, zhouyong\}@shanghaitech.edu.cn}, he.chen@ie.cuhk.edu.hk}
}
 
\maketitle

\begin{abstract}
We consider the multicast transmission of a real-time Internet of Things (IoT) system, where a server transmits time-stamped status updates to multiple IoT devices. 
We apply a recently proposed metric, named age of information (AoI), to capture the timeliness of the information delivery.
The AoI is defined as the time elapsed since the generation of the most recently received status update.
Different from the existing studies that considered either multicast transmission without hard deadlines or unicast transmission with hard deadlines, we enforce a hard deadline for the service time of multicast transmission.  
This is important for many emerging multicast IoT applications, where the outdated status updates are useless for IoT devices. 
Specifically, the transmission of a status update is terminated when either the hard deadline expires or a sufficient number of IoT devices successfully receive the status update. 
We first calculate the distributions of the service time for all possible reception outcomes at IoT devices, and then derive a closed-form expression of the average AoI. 
Simulations validate the performance analysis, which reveals that: 1) the multicast transmission with hard deadlines achieves a lower average AoI than that without hard deadlines; and 2) there exists an optimal value of the hard deadline that minimizes the average AoI.   
\end{abstract}

\section{Introduction}
Internet of Things (IoT) provides ubiquitous wireless connectivity and automated information delivery for massive devices that have the capabilities of monitoring, processing, and communication \cite{taleb2012machine}. 
Many emerging IoT applications, including environment monitoring, smart metering, and autonomous transportation, require timely information delivery, i.e., the statuses observed at the receivers are always fresh. 
The conventional performance metrics (e.g., throughput and delay) cannot adequately capture the information freshness.
Fortunately, the information freshness can be well characterized by a recently proposed performance metric termed as \textit{Age of Information} (AoI) \cite{statusqueue}.
The AoI at a receiver is defined as the time difference between the current time and the generation time of the most recently received status update.

The analysis and optimization of the AoI performance for various systems have recently attracted considerable attention \cite{real, ontheage, yates2012real, optimizing, zhou2019joint, zhou2019minimum, gu2019timely}.
Relying on queueing theory, the authors in \cite{real} analyzed the average AoI for different queueing models.
Results in \cite{real} showed that minimizing the AoI is different from minimizing the delay, because the delay does not capture the inter-delivery time of status updates.
This seminal analysis was further extended in \cite{ontheage} and \cite{yates2012real} to show that the average AoI can be decreased by reducing the buffer size and increasing the server number, respectively.
With heterogeneous service time distributions, the authors in \cite{optimizing} derived the average AoI for an $M/G/1$ queueing system.
% Moreover, the average AoI was also analyzed and optimized for networks with Markov channels \cite{markov} and energy harvesting devices \cite{arafa2018age}.
The authors in \cite{zhou2019joint} proposed age-minimal scheduling schemes that jointly optimize the status sampling and updating for IoT networks. 
The authors in \cite{zhou2019minimum} studied the minimum AoI with non-uniform status packet sizes in IoT networks.
Besides, the tradeoff between AoI and energy efficiency of IoT monitoring systems was characterized  in \cite{gu2019timely}, where a limited number of retransmissions are allowed for each status update.
We note that these studies mainly focused on the status update systems with unicast transmission. 

Multicast transmission can simultaneously serve multiple receivers that are interested in the same information. 
The timely delivery of these popular information requested by multiple IoT devices is critical for many emerging IoT applications.
For example, in a connected vehicular network, the status update of an autonomous vehicle, especially the safety warning message, needs to be timely disseminated to the nearby vehicles and pedestrians. 
In a smart parking lot, a server continuously collects the occupancy information of all parking spaces and reports the locations of the vacant parking spaces to the nearby drivers. 
The authors in \cite{multicast, status} derived the average AoI of a multicast system, where a status update is dropped as soon as it has been successfully received by a certain number of receivers.
The tradeoff between the energy efficiency and the average AoI in multicast systems was studied in \cite{optimum}, where a scheduling strategy based on the optimal stopping theories was proposed.
In \cite{Buyukates2018Age}, the authors studied the average AoI in a two-hop multicast network. 
In addition, the authors in \cite{schedulingbroadcast} proposed several scheduling policies to minimize the average AoI for broadcast transmission over unreliable channels.
However, the aforementioned studies on multicast transmission did not take into account the hard deadline. 
This is crucial for many real-time IoT applications, where the outdated status updates have little value to be delivered. 
It has been demonstrated in \cite{on, analysis} that the packet deadline has a significant impact on the average AoI for unicast transmission.
Specifically, the authors in \cite{on, analysis} derived the closed-form expressions of the average AoI for $M/M/1$ and $M/G/1$ queueing systems, respectively, where the waiting time of each packet is subject to a hard deadline but the service time can be arbitrary large. 
Motivated by the aforementioned emerging applications and existing studies, we are interested in studying the average AoI of multicast transmission with hard deadlines, which remains unexplored to the best of our knowledge.

%Motivated by the observation in \cite{status}, we study the AoI that enforce deadlines in multicast networks. Extending the results of \cite{status}, we analyze the evolution of average AoI after enforce a hard deadline.
% {\color{blue} In \cite{analysis},  both packet deadlines and buffer capacity are considered,  the system modeled as $M/M/1+G$ and $M/M/1+D$ queue, where the last symbol respresents the distribution of deadlines, and they derived the expressions for each case. }

%The effect of buffer size and packet deadline on AoI was studied in \cite{analysis}-\cite{controlling}, it was shown that AoI is very easily affected by the deadline in queueing system. In \cite{on},  the authors analyze the system for the cases of a fixed deadline and a random random deadline and generalized the available results for the average age in $M/M/1$ queueing system. It was shown that the use of a packet deadline can add a new dimension to optimizing the age performance.

In this paper, we analyze the average AoI of a real-time IoT system, where a server multicasts information to multiple IoT devices.
Different from the existing studies that considered either multicast transmission without hard deadlines \cite{multicast} or unicast transmission with hard deadlines \cite{on}, we enforce a hard deadline for the service time of multicast transmission.
Once a status update is generated, it is time-stamped and transmitted by the server.
The server terminates the transmission of a status update if either the deadline expires or a sufficient number of IoT devices successfully receive the status update.
It is worth noting that the instantaneous AoI evolution is more complicated for multicast transmission with hard deadlines than that for the existing studies \cite{status,on}.
This is because, the instantaneous AoI evolution in this paper depends on both the reception outcomes of multiple IoT devices and the hard deadline, whereas that in the existing studies only depends on one of these important factors. 
%{\color{red}This is because, for conventional schemes, the transmission of a status update is terminated only when $K$ IoT devices successfully receive the status update in multicast networks or the deadline expires in queuing system. At this moment, after the server transmits the status update at each time, the amount of instantaneous AoI increases is the service time of $K$-th device and the deadline, respectively. But in this paper it has two possibilities after each transmission i.e., the deadline or the $K$-th devices' service time.}
We explicitly show that the average AoI depends on the service time of multiple devices, the hard deadline, and the total number of IoT devices.
The main contributions of this paper are summarized as follows.

%{\color{blue}
%In unicast systems, after the node receives the update, the source generates next update immediately. However, in multicast systems, the node still has to wait for others nodes to receive update.}
\begin{itemize}
\item We derive the probability density functions (PDFs) of the service time by using order statistics for all possible reception outcomes at the receiving IoT devices, and obtain the first and second moments of the inter-generation time of two consecutive status updates. 

\item We derive the closed-form expression of the average AoI for multicast transmission with hard deadlines, which includes multicast transmission without hard deadlines, broadcast transmission with hard deadlines, and unicast transmission with hard deadlines as special cases.

\item Simulations validate the theoretical analysis, which illustrates the impact of various parameters on the average AoI. 
Results also reveal that the average AoI of multicast transmission with hard deadlines is lower than that without hard deadlines, and the deadline can be further optimized to minimize the average AoI. 
\end{itemize}

The rest of this paper is organized as follows. In Section \ref{model}, we describe the system model and the AoI evolution. The average AoI of multicast transmission with hard deadlines is analyzed in Section \ref{k}. The numerical results are presented in Section \ref{sim}. Finally, Section \ref{con} concludes this paper.

\section{System Model} \label{model}
% \begin{figure}
% 	\centering
% 	\includegraphics[scale=0.6]{model.eps}
% 	\caption{A server multicasts status updates with deadline $T_{\mathrm{D}}$ to multiple nodes. Some nodes (e.g., node 3) cannot successfully receive status update $j$ before its transmission is terminated by the server. Note that $T_N(K)$ denotes the time duration that $K$ nodes have successfully received the status update. }
% 	\label{fig:system}
% \vspace{-3mm}
% \end{figure}

Consider a real-time IoT system consisting of a single server transmitting multicast information to $N$ IoT devices.
We denote the sets of status updates and IoT devices as $\mathcal{J} = \{1, \ldots, j, \ldots\}$ and $\mathcal{N} = \{1, \ldots, n, \ldots,N \}$, respectively. 
The status updates are generated by the server and there is no random status update arrival. 
Each status update is time-stamped and transmitted by the server once it is generated.
The time required to successfully deliver status update $j$ to IoT device $n$ is denoted as $T_{n,j}$.
We follow \cite{multicast, status} and assume that $\{T_{n,j}\}$ are independent and identically distributed (i.i.d.) shifted exponential random variables with rate $\lambda_s$ and positive constant shift $c$. 
Hence, the PDF of $T_{n,j}$ is $f_{T}(t) = \lambda_s \mathrm{e}^{- \lambda_s (t-c)}$, where $t > c$. 
A status update is considered to be \textit{served} when it is successfully received by at least $K$ IoT devices, where $K \le N$.
Once a device successfully receives a status update, it sends an acknowledgement (ACK) packet back to the server via an error-free and delay-free control channel.
In addition, we consider that each status update subjects to a hard deadline, which is denoted as $T_{\mathrm{D}}$.
Specifically, if a status update is not successfully received by $K$ intended devices when the deadline expires, then this status update is considered to be useless to the IoT devices and immediately \textit{dropped} by the server. 
The server \textit{terminates}  the transmission of the current status update (e.g., $j$) if it is either served or dropped, and meanwhile generates and time-stamps a new status update (e.g., $j+1$).

The instantaneous AoI of device $n$ at time $t$ is defined as $\Delta_n(t) = t - u_n(t)$, where $u_n(t)$ denotes the generation time of the most recently received status update at device $n$ as of time $t$.
We depict the evolution of the instantaneous AoI at device $n$ over time as the sawtooth pattern in Fig. \ref{fig:aoi}.
As can be observed, the instantaneous AoI increases linearly with time $t$ and drops to a smaller value once a new status update containing fresher information is received.
 
\begin{figure}
	\centering
	\includegraphics[scale=0.75]{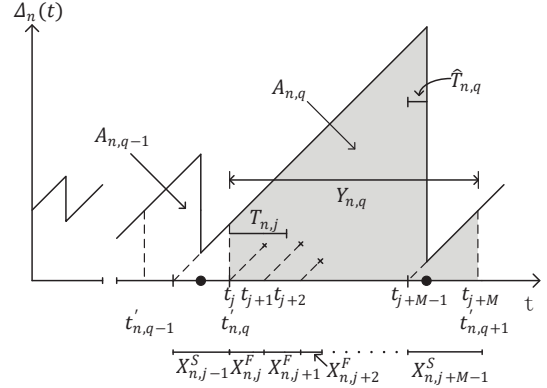}
		\vspace{-2mm}
	\caption{Age evolution of device $n$ over time. The time instances that device $n$ successfully receive status updates are marked by $\bullet$.}
	\label{fig:aoi}
	 \vspace{-6mm}
\end{figure}

To better describe the AoI evolution, we present the following definitions.
We denote $t_j$ as the time instant that the server generates status update $j \in \mathcal{J}$.
We define $X_{n,j}^{\mathrm{F}} = t_{j+1} - t_j $ as the inter-generation time of two consecutive status updates $j$ and $j+1$ if status update $j$ fails to be received by device $n$.
Similarly, we define $X_{n,j}^{\mathrm{S}} = t_{j+1} - t_j$ as the inter-generation time of two consecutive status updates $j$ and $j+1$ if status update $j$ is successfully received by device $n$.
%We define {\color{blue}  $X_{n,j}^\mathrm{F} = t_{j+1} - t_j$ as the wait time required when the node $n$ fail to  receive the update j. Define $X_{n,j+M-1}^\mathrm{S} = t_{j+M+1} - t_{j+M}$ as the wait time for the node $n$ to successfully receive an update. }
Due to the hard deadline (i.e., $T_{\mathrm{D}}$) and the randomness of service time $T_{n,j}$, it is possible that some status updates cannot be successfully received by device $n$.
Thus, we further denote $t_{n,q}'$  as the termination time of the $(q-1)$-th status update that has been successfully received by device $n$.
As shown in {Fig. \ref{fig:aoi}}, $t_{n,q}'=t_j$ implies that the $(j-1)$-th status update transmitted by the server is the $(q-1)$-th status update successfully received by device $n$, where $j \ge q$.
Note that we use $j$ and $q$ to index the status updates transmitted by the server and successfully received by the IoT device, respectively. 
% i.e., corresponding to the $(j+M)$-th status update transmitted by the server.

As $\{T_{n,j}, n \in \mathcal{N}, j \in \mathcal{J}\}$ are i.i.d., the evolution processes of the instantaneous AoI  for all IoT devices are statistically identical and hence each device ends up having the same average AoI. 
This allows us to focus on analyzing the average AoI of any device.
We denote $\mathcal{Q(\mathcal{T})}=\max\{q|t_{n,q}'\le \mathcal{T}\}$ as the number of status updates that have been successfully received by device $n$ as of time $\mathcal{T}$.
As in \cite{real}, the average AoI of device $n$ is given by 
\begin{equation}
	 \begin{split}
\bar{\Delta}_n&=\lim\limits_{\mathcal{T} \to \infty }{\frac{1}{\mathcal{T}}\int_{0}^{\mathcal{T}} \Delta_n(t)} \, dt \\
&=\lim\limits_{\mathcal{T} \to \infty }\frac{Q(\mathcal{T})}{\mathcal{T}}\frac{1}{Q(\mathcal{T})}\sum_{q=1}^{Q(\mathcal{T})} A_{n,q}=\frac{\mathbb{E}[A_{n,q}]}{\mathbb{E}[Y_{n,q}]},
\label{ex3}
\end{split}
\end{equation}
where $\frac{Q(\mathcal{T})}{\mathcal{T}}$ is the steady-state rate of the update delivery, $A_{n,q}$ is the area of the shaded polygon under the sawtooth curve in {Fig. \ref{fig:aoi}}, and $Y_{n,q}=t_{n,q+1}' - t_{n,q}'$ denotes the time duration starting from the termination time of the $(q-1)$-th update to that of the $q$-th update at device $n$. 
Based on Fig. \ref{fig:aoi}, we have
\begin{equation}
A_{n,q}=(X_{n,j-1}^\mathrm{S}+W)\hat{T}_{n,q}+{(2X_{n,j-1}^\mathrm{S}+W)\frac{W}{2}}+\frac{1}{2} \left({X_{n,j+M-1}^\mathrm{S}}\right)^2,
\label{ex4}
\end{equation}
where $\hat{T}_{n,q}$ denotes the service time of the $q$-th status update successfully delivered to device $n$, $M$ is the number of status updates transmitted by the server during $\left[t_{n,q}', t_{n,q+1}'\right)$, and $W=\sum_{i=j}^{j+M-2}X_{n,i}^\mathrm{F}$. As $\{X_{n,j}^\mathrm{F}, j \in \mathcal{J}\}$ are i.i.d., we denote  $\mathbb{E}[X_{n,j}^\mathrm{F}] = \mathbb{E}[X_{n}^\mathrm{F}]$.
As $X_{n,j-1}^\mathrm{S}$, $\hat{T}_{n,q}$, and $X_{n,j}^\mathrm{F}$ are independent of each other, the expectation of $A_{n,q}$ can be expressed as
\begin{equation}
\hspace{-0mm}	 \begin{split}
	&\mathbb{E}[A_{n,q}]  =\mathbb{E}[X_{n,j-1}^\mathrm{S}]\mathbb{E}[\hat{T}_{n,q}]+\mathbb{E}[W]\mathbb{E}[\hat{T}_{n,q}] \\
	&+ \mathbb{E}[X_{n,j-1}^\mathrm{S}]\mathbb{E}[W]+\frac{1}{2}\mathbb{E}\left[W^2\right]+\frac{1}{2}\mathbb{E}\left[ \left({X_{n,j+M-1}^\mathrm{S}}\right)^2\right]. 
\label{ex6}
\end{split}
\end{equation}

\noindent As $X_{n,j-1}^\mathrm{S}$ and $X_{n,j+M-1}^\mathrm{S}$ are identically distributed, we denote $\mathbb{E} \left[ X_{n,j-1}^\mathrm{S} \right] = \mathbb{E} \left[ X_{n,j+M-1}^\mathrm{S} \right]=\mathbb{E} \left[ X_{n}^\mathrm{S} \right]$ and rewrite \eqref{ex6} as
\begin{equation}
	 \begin{split}
  \mathbb{E}[A_{n,q}]= &\frac{1}{2} \mathbb{E}\left[W^2\right] + \left( \mathbb{E}\left[\hat{T}_{n,q}\right] + \mathbb{E}\left[X_{n}^\mathrm{S}\right] \right) \mathbb{E}[W]  \\
	&+\mathbb{E}\left[X_{n}^\mathrm{S}\right]\mathbb{E}\left[\hat{T}_{n,q}\right] + \frac{1}{2}{\mathbb{E}\left[{\left(X_{n}^\mathrm{S}\right)}^2\right]}.
\label{ex5}
\end{split}
\end{equation}

The time duration of the shaded polygon can also be written as $Y_{n,q}=W+X_{n,j+M-1}^\mathrm{S}$, and hence its expectation is
\begin{equation}
  \mathbb{E}[Y_{n,q}]=\mathbb{E}[W]+\mathbb{E}\left[X_{n}^\mathrm{S}\right].
\label{ex7}
\end{equation}

By substituting (\ref{ex5}) and (\ref{ex7}) into (\ref{ex3}), we have
\begin{equation}
 \begin{split}
	 \bar{\Delta}_n =   \frac{\mathbb{E}[A_{n,q}]}{\mathbb{E}[Y_{n,q}]} = &\frac{\mathbb{E}\left[W^2\right]  + 2 \left( \mathbb{E}\left[\hat{T}_{n,q}\right] + \mathbb{E}\left[X_{n}^\mathrm{S}\right] \right) \mathbb{E}[W]  }{2\mathbb{E}[W]+2\mathbb{E}\left[X_{n}^\mathrm{S}\right]}  \\
	  &+\frac{2 \mathbb{E}\left[X_{n}^\mathrm{S}\right]\mathbb{E}\left[\hat{T}_{n,q}\right] + \mathbb{E}\left[{\left(X_{n}^\mathrm{S}\right)}^2\right]  }{2\mathbb{E}[W]+2\mathbb{E}\left[X_{n}^\mathrm{S}\right]}.
	 \label{AveAoI}
 \end{split}
\end{equation}

To obtain the closed-form expression of the average AoI, we calculate all the expectation terms in \eqref{AveAoI} in Section \ref{k}.
% \begin{eqnarray}
% \bar{\Delta}_n =   \frac{\mathbb{E}[A_{n,q}]}{\mathbb{E}[Y_{n,q}]} = \frac{\mathbb{E}\left[W^2\right]  + 2 \left( \mathbb{E}\left[\hat{T}_{n,q}\right] + \mathbb{E}\left[X_{n}^\mathrm{S}\right] \right) \mathbb{E}[W] + 2 \mathbb{E}\left[X_{n}^\mathrm{S}\right]\mathbb{E}\left[\hat{T}_{n,q}\right] + \mathbb{E}\left[{\left(X_{n}^\mathrm{S}\right)}^2\right]  }{2\mathbb{E}[W]+2\mathbb{E}\left[X_{n}^\mathrm{S}\right]}.
% \label{AveAoI}
% \end{eqnarray}

%
%We will derive the terms $\mathbb{E}[X_{n,j-1}^\mathrm{S}]$, $\mathbb{E}[\hat{T}_{n,q}]$, $\mathbb{E}[X_{n}^\mathrm{F}]$, $\mathbb{E}[W]$, $\mathbb{E}[W^2]$, $\mathbb{E}[B_{n,j+M+1}]$, $\mathbb{E}[B_{n,j+M+1}^2]$, $\mathbb{E}[Y_{n,q}]$ in the following sections.

\section{Analysis of Average AoI}\label{k}
In this section, we calculate the expressions of $\mathbb{E}\left[ X_{n}^{\mathrm{F}} \right]$, $\mathbb{E} \left[ \left(X_n^{\mathrm{F}} \right)^2 \right]$, $\mathbb{E}\left[X_n^{\mathrm{S}}\right]$, $\mathbb{E}\left[\left(X_n^{\mathrm{S}}\right)^2\right]$, $\mathbb{E}[W]$, $\mathbb{E}[W^2]$, and $\mathbb{E}[\hat{T}_{n,q}]$, based on which we derive the average AoI given in \eqref{AveAoI}. 

%$\mathbb{E}\left[X_{n}^\mathrm{F}\right]$, $\mathbb{E}\left[{\left(X_{n}^\mathrm{F}\right)}^2\right]$, $\mathbb{E}\left[X_{n}^\mathrm{S}\right]$,  $\mathbb{E}\left[{\left(X_{n}^\mathrm{S}\right)}^2\right]$, $\mathbb{E}\left[W\right]$, $\mathbb{E}\left[W^2\right]$,  and $\mathbb{E}[\hat{T}_{n,q}]$.

\subsection{First and Second Moments  of Inter-Generation Time $X_n^{\mathrm{F}}$} \label{SubSec_XFN}
We first calculate the expectation of the inter-generation time of two consecutive status updates when the former status update is not successfully received by device $n$, i.e., $\mathbb{E}[X_{n}^\mathrm{F}]$.
Recall that the server terminates the transmission of a status update when one of the following two events occurs: 1) Event I - The deadline of the status update expires; 2) Event II - At least $K$ devices successfully receive the status update ahead of device $n$. Thus, device $n$ fails to receive the status update if $T_{n,j} > \min\{T_{\mathrm{D}}, T_{N}(K) \}$, where $T_{N}(K)$ is defined as the time duration that $K$ devices have successfully received the status update and it is the $K$-th smallest variable in set $\{T_{n,j}, n \in \mathcal{N}\}$. 
Based on order statistics \cite{order}, the PDF of $T_{N}(K)$ is given by 
\begin{eqnarray} \label{eqn:tnk9}
f_{T_{N}(K)}(t) = K\binom{N}{K} \left(F_{T}(t)\right)^{K-1} \left(1-F_{T}(t) \right)^{N-K}  f_{T}(t),
\end{eqnarray}
where $F_{T}(t) = 1 - \mathrm{e}^{- \lambda_s (t-c)}$ with $t > c$ is the cumulative distribution function (CDF) of $T_{n,j}$.

We denote the case that device $n$ fails to receive the status update as $\mathcal{C}_{\mathrm{F}}$.
When $T_{n,j} > \min\{T_{\mathrm{D}}, T_{N}(K) \}$, due to the randomness of service times, $X_n^{\mathrm{F}}$  behaves differently for the following two cases: (1) $\mathcal{C}_{\mathrm{F},1}$ - Event II occurs earlier than Event I (i.e., $T_{N}(K) < T_{\mathrm{D}}$); (2) $\mathcal{C}_{\mathrm{F},2}$ - Event I occurs earlier than Event II (i.e., $T_{\mathrm{D}} < T_{N}(K)$).
When Case $\mathcal{C}_{\mathrm{F},1}$ occurs, the instantaneous AoI of device $n$ increases by $T_{N}(K)$ (i.e., $X_{n}^\mathrm{F} = T_{N}(K)$).
On the other hand, when Case $\mathcal{C}_{\mathrm{F},2}$ occurs, the instantaneous AoI of device $n$ increases by $T_{\mathrm{D}}$ (i.e., $X_{n}^\mathrm{F} = T_{\mathrm{D}}$). Thus, the expectation of inter-generation time $X_{n}^{\mathrm{F}}$ is given by
\begin{eqnarray}
\mathbb{E}\left[ X_{n}^{\mathrm{F}} \right] = \mathbb{P} \left( \mathcal{C}_{\mathrm{F},1}  \right) \mathbb{E}\left[T_N(K) \left| \mathcal{C}_{\mathrm{F},1}\!\right.\right] + \mathbb{P} \left( \mathcal{C}_{\mathrm{F},2}  \right) T_{\mathrm{D}},
\label{eqn_xnf7}
\end{eqnarray}
where $\mathrm{P}(\mathcal{C}_{\mathrm{F},1})$ and $\mathrm{P}(\mathcal{C}_{\mathrm{F},2})$ denote the probabilities that Cases $\mathcal{C}_{\mathrm{F},1}$ and $\mathcal{C}_{\mathrm{F},2}$ occur when device $n$ fails to receive the status update, respectively, with $\mathrm{P}(\mathcal{C}_{\mathrm{F},1}) + \mathrm{P}(\mathcal{C}_{\mathrm{F},2}) = 1$.
%where $\mathbb{E}\left[T_N(K) \left| \mathcal{C}_{\mathrm{F},1}\!\right.\right]$ is given in (\ref{eqn8}).
Similarly, the second moment of inter-generation time $X_{n}^{\mathrm{F}}$ is given by 
\begin{eqnarray}
\mathbb{E} \left[ \left(X_n^{\mathrm{F}} \right)^2 \right] = \mathbb{P} \left( \mathcal{C}_{\mathrm{F},1} \right) \mathbb{E}[T_N^2(K)| \mathcal{C}_{\mathrm{F},1} ] + \mathbb{P} \left( \mathcal{C}_{\mathrm{F},2}  \right) T_{\mathrm{D}}^2.
 \label{ex17}
\end{eqnarray}
%where $\mathbb{E}\left[T_N^2(K)| \mathcal{C}_{\mathrm{F},1} \right]$ is given in (\ref{eqn9}).
%We calculate $\mathbb{P} \left( \mathcal{C}_{\mathrm{F},1} \right)$,
%$\mathbb{E}\left[T_N(K) \left| \mathcal{C}_{\mathrm{F},1}\!\right.\right] $,
% $ \mathbb{P} \left( \mathcal{C}_{\mathrm{F},2} \right)$,
% and $\mathbb{E}[T_N^2(K)| \mathcal{C}_{\mathrm{F},1} ]$ as follows.

To calculate (\ref{eqn_xnf7}) and (\ref{ex17}), we first analyze the first and second moments of conditional $T_N(K)$, i.e., $\mathbb{E}\left[T_N(K) \left| \mathcal{C}_{\mathrm{F},1} \right. \right]$ and $\mathbb{E}\left[T_N^2(K) \left| \mathcal{C}_{\mathrm{F},1} \right. \right]$, in the following proposition. 

\begin{prop}
The first and second moments of the time duration that $K$ devices successfully receive a status update (i.e., $T_{N}(K)$) conditioning on the occurrence of Case $\mathcal{C}_{\mathrm{F},1}$ are
\begin{equation}
	\begin{split}
 	\mathbb{E}\left[T_N(K) \left| \mathcal{C}_{\mathrm{F},1} \right. \right]  & = \frac{1}{1 - \mathcal{Z}_{K}} \sum_{j=0}^{K-1} \frac{B_{K,j} }{\lambda_s U_{K,j}^2}\Big[1+c\lambda_s U_{K,j}  \\
	& \hspace{3mm} - \, (1 + T_{\mathrm{D}}\lambda_s U_{K,j}  )V_{K,j}\Big],
		 \label{Eq_TNK1}
	\end{split}
\end{equation}
\begin{equation}
	\begin{split}
	\mathbb{E}\left[T_N^2(K) \left| \mathcal{C}_{\mathrm{F},1} \right. \right] & = \frac{1}{1-\mathcal{Z}_{K}} \sum_{j=0}^{K-1} \frac{B_{K,j}}{\lambda_s^2 U_{K,j}^3} \Big[ (1 + c \lambda_s U_{K,j})^2  \\
	& \hspace{3mm} + 1  - \left((1+T_{\mathrm{D}}\lambda_s U_{K,j}  )^2 +1\right)V_{K,j} \Big],
	\label{eqn9}
	\end{split}
\end{equation}
where $B_{K,j} = K\binom{N}{K} \binom{K-1}{j} (-1)^j $, $U_{K,j}=N-K+1+j$, $V_{K,j}=e^{-\lambda_s U_{K,j}(T_{\mathrm{D}}-c)}$, and 
$\mathcal{Z}_{K} = \mathbb{P}(T_{\mathrm{D}}<T_N(K))= \sum_{i=0}^{K-1} B_{K,i}\frac{V_{K,i}}{U_{K,i}}$.
\end{prop}
\begin{proof}
See Appendix A.
\end{proof}

The occurrence probability of Case $\mathcal{C}_{\mathrm{F},2}$ is given in the following proposition. 
\begin{prop}
The probability that Case $\mathcal{C}_{\mathrm{F},2}$ occurs can be expressed as 
	\begin{eqnarray} 
		\label{Eq_CF}
		\label{pcf2} \mathbb{P}(\mathcal{C}_{\mathrm{F},2})  = \frac{ (N-K)\mathcal{Z}_K + N\sum_{h=1}^{K}\mathcal{Z}_h } { N e^{-\lambda_s(T_{\mathrm{D}}-c)}+(N-K)- \sum_{h=K+1}^{N}\mathcal{Z}_h},
	\end{eqnarray}
\end{prop}
\begin{proof}
	See Appendix B.
\end{proof}

By definition, we have $\mathbb{P}(\mathcal{C}_{\mathrm{F},1}) = 1 - \mathbb{P}(\mathcal{C}_{\mathrm{F},2})$. 
By substituting (\ref{Eq_TNK1}) -- (\ref{pcf2}) into 
(\ref{eqn_xnf7}) and (\ref{ex17}),
we obtain $\mathbb{E}\left[ X_{n}^{\mathrm{F}} \right]$ and $\mathbb{E} \left[ \left(X_n^{\mathrm{F}} \right)^2 \right]$.

%Fig. \ref{fig:flow} show that age evolution under all cases. On the one hand, we consider node $n$ fail to receive update, i.e., case $I$ occurs, $T_{\mathrm{D}} < T_{n,j}$ or $T_{N}(K) < T_{n,j} $. Then, if $T_{\mathrm{D}} < T_{N}(K)$ , i.e., case $I_a$ occurs, AoI adds $T_{\mathrm{D}}$. Otherwise, case $I_a$ occurs, AoI adds $T_{N}(K)$.
%On the other hand, we consider node $n$ receive update successfully, i.e., case $II$ occurs, $T_{n,j} \le T_{\mathrm{D}}$ and $T_{n,j} \le T_{N}(K)$. In this time AoI reset to a small value. Then, node $n$ still need to wait the deadline expires or remaining of $k$ nodes receive update. If $T_{\mathrm{D}} < T_{N}(K)$ , i.e., case $II_a$ occurs, AoI continues to adds $T_{\mathrm{D}}$. Otherwise, case $II_a$ occurs, AoI continues to adds $T_{N}(K)$.  By the way, if
%$T_{N}(K) = T_{n,j} $, node $n$ does not need to wait.

%\begin{figure}
%	\centering
%	\includegraphics[scale=0.5]{flow.eps}
%	\caption{AoI evolution with all cases.}
%	\label{fig:flow}
%\end{figure}

\subsection{First and Second Moments of Inter-Generation Time $X_n^{\mathrm{S}}$}
In this subsection, we derive the first and second moments of the inter-generation time of two consecutive status updates when the former status update is successfully received by device $n$, i.e., $\mathbb{E}\left[ X_n^{\mathrm{S}} \right]$ and $\mathbb{E}\left[ \left( X_n^{\mathrm{S}} \right)^2 \right]$.

Note that device $n$ successfully receives status update $j$ if $ T_{n,j} \le \min \{ T_{\mathrm{D}}, T_{N}(K) \}$.
We denote the case that device $n$ successfully receives the status update as $\mathcal{C}_{\mathrm{S}}$.
We observe that $X_{n}^{\mathrm{S}}$ behaves differently for the following two cases: (1) $\mathcal{C}_{\mathrm{S},1}$ - Event II occurs earlier than Event I (i.e., $T_{N}(K) < T_{\mathrm{D}}$); (2) $\mathcal{C}_{\mathrm{S},2}$ - Event I occurs earlier than Event II (i.e., $T_{\mathrm{D}} < T_N(K)$).
When Case $\mathcal{C}_{\mathrm{S},1}$ occurs, the instantaneous AoI of device $n$ increases by $T_{N}(K)$ (i.e., $X_n^{\mathrm{S}} = T_{N}(K)$).
When Case $\mathcal{C}_{\mathrm{S},2}$ occurs, the instantaneous AoI of device $n$ increases by $T_{\mathrm{D}}$ (i.e., $X_n^{\mathrm{S}} = T_{\mathrm{D}}$).
The first and second moments of $\mathrm{E}[X_{n}^\mathrm{S}]$ are given by
\begin{eqnarray}
\label{ex30_1} \mathbb{E}[X_{n}^\mathrm{S}] && \hspace{-6mm} =\mathbb{P}(\mathcal{C}_{\mathrm{S},1}) \mathbb{E}[T_N(K)|\mathcal{C}_{\mathrm{S},1}]+\mathbb{P}(\mathcal{C}_{\mathrm{S},2})T_{\mathrm{D}}, \\
\mathbb{E}\left[\left(X_{n}^\mathrm{S}\right)^2\right] &&\hspace{-6mm} =\mathbb{P}(\mathcal{C}_{\mathrm{S},1}) \mathbb{E}[T_N^2(K)|\mathcal{C}_{\mathrm{S},1}]+\mathbb{P}(\mathcal{C}_{\mathrm{S},2})T_{\mathrm{D}}^2,
\label{ex30}
\end{eqnarray}
where $\mathbb{P} \left( \mathcal{C}_{\mathrm{S},1} \right)$ and $\mathbb{P} \left( \mathcal{C}_{\mathrm{S},2} \right)$ denote the probabilities of the occurrence of Cases $\mathcal{C}_{\mathrm{S},1}$ and $\mathcal{C}_{\mathrm{S},2}$ when device $n$ successfully receives the status update, respectively, with $\mathbb{P} \left( \mathcal{C}_{\mathrm{S},1} \right) + \mathbb{P} \left( \mathcal{C}_{\mathrm{S},2} \right) = 1$.
To obtain $\mathbb{E}[X_{n}^\mathrm{S}]$ and $\mathbb{E}\left[\left(X_{n}^\mathrm{S}\right)^2\right]$, we need to calculate $\mathbb{P} \left( \mathcal{C}_{\mathrm{S},1} \right)$, $\mathbb{E}[T_N(K)|\mathcal{C}_{\mathrm{S},1}]$, and $\mathbb{E}[T_N^2(K)|\mathcal{C}_{\mathrm{S},1}]$.
The following proposition gives the first and second moments of $T_N(K)$ conditioning on the occurrence of Case $\mathcal{C}_{\mathrm{S},1}$.
\begin{prop}
The first and second moments of the time that $K$ IoT devices successfully receive a status update (i.e., $T_N(K)$) conditioning on the occurrence of Case $\mathcal{C}_{\mathrm{S},1}$ are given by
$ \mathbb{E}[T_N(K)|\mathcal{C}_{\mathrm{S},1}] = \mathbb{E}\left[T_N(K) \left| \mathcal{C}_{\mathrm{F},1} \right. \right] $ and $\mathbb{E}[T_N^2(K)|\mathcal{C}_{\mathrm{S},1}] = \mathbb{E}\left[T_N^2(K) \left| \mathcal{C}_{\mathrm{F},1} \right. \right]$, respectively,
%\begin{equation}
% \begin{split}
% \mathbb{E}[T_N(K)|\mathcal{C}_{\mathrm{S},1}] = \frac{\Theta_1}{\Theta_3}, \\
% \mathbb{E}[T_N^2(K)|\mathcal{C}_{\mathrm{S},1}] = \frac{\Theta_2}{\Theta_3},
% \end{split}
% \label{ex52}
%\end{equation}
where $\mathbb{E}\left[T_N(K) \left| \mathcal{C}_{\mathrm{F},1} \right. \right] $ and $\mathbb{E}\left[T_N^2(K) \left| \mathcal{C}_{\mathrm{F},1} \right. \right]$ are given in Proposition 1.
% \begin{eqnarray}
%   \mathbb{E} \left[ T_N(K) \left| \mathcal{C}_{\mathrm{S},1} \right. \right]&& \hspace{-6mm} =\int_{0}^{T_{\mathrm{D}}}tF'_{T_N(K)|\mathcal{C}_{\mathrm{S},1}}(t) \nonumber \\
%     && \hspace{-6mm} =\left(K\binom{N}{k} \sum_{j=0}^{K-1}\binom{K-1}{j} (-1)^j \frac{1-e^{-\lambda_s\mathcal{R}T_{\mathrm{D}}}(\lambda_s\mathcal{R}T_{\mathrm{D}}+1)}{\mathcal{R}^2}\right)/ \nonumber \\
%     && \hspace{-6mm} \left(1-K\binom{N}{K} \sum_{j=0}^{K-1}\binom{K-1}{j} (-1)^j \frac{e^{-\lambda_s\mathcal{R}T_{\mathrm{D}}}}{\mathcal{R}}\right), \\
%
% \mathbb{E} \left[ T_N^2(K) \left| \mathcal{C}_{\mathrm{S},1} \right. \right] && \hspace{-6mm} =  \int_{0}^{T_{\mathrm{D}}}t^2F'_{T_N(K)|\mathcal{C}_{\mathrm{S},1}}(t) \nonumber\\
%    && \hspace{-6mm} k\binom{N}{K} \sum_{j=0}^{K-1}\binom{K-1}{j} (-1)^j \nonumber\\
%    && \hspace{-6mm} \frac{2-e^{-\lambda_s\mathcal{R}T_{\mathrm{D}}}((\lambda_s\mathcal{R}T_{\mathrm{D}})^2+2\lambda_s\mathcal{R}T_{\mathrm{D}}+2)}{\lambda_s^2\mathcal{R}^3} /\nonumber\\
%   && \hspace{-6mm} \left(1-K\binom{N}{K} \sum_{j=0}^{K-1}\binom{K-1}{j} (-1)^j \frac{e^{-\lambda_s\mathcal{R}T_{\mathrm{D}}}}{\mathcal{R}}\right).
% \end{eqnarray}
\end{prop}

\begin{proof}
The proof can be readily obtained by following the steps as in Appendix A. Due to space limitation, the detailed proof is omitted.
\end{proof}
%  where $\Phi_2=K\binom{N}{K} \sum_{j=0}^{K-1}\binom{K-1}{j} (-1)^j \frac{2-e^{-\lambda_s\mathcal{R}T_{\mathrm{D}}}((\lambda_s\mathcal{R}T_{\mathrm{D}})^2+2\lambda_s\mathcal{R}T_{\mathrm{D}}+2)}{\lambda_s^2\mathcal{R}^3}$.

By definition, the occurrence probability of Case $\mathcal{C}_{\mathrm{S},1}$ can be calculated by
\begin{equation}
\begin{split}
\mathbb{P}(\mathcal{C}_{\mathrm{S},1}) =  \frac{K(1 - \mathcal{Z}_{K})}{N\mathbb{P}(\mathcal{C}_{\mathrm{S}})},
\label{ex53}
\end{split}
\end{equation}
where $\mathbb{P}(\mathcal{C}_{\mathrm{S}})$ denotes the probability that device $n$ successfully receives the status update and can be calculated by
$\mathbb{P}(\mathcal{C}_{\mathrm{S}})=\mathbb{P}( T_{n,j} < \min \{ T_{\mathrm{D}}, T_N(K)\})
 = \frac{1}{N} \sum_{h=1}^{K}   \left(1 - \mathcal{Z}_{h}\right)$.

By substituting the derived expressions of $\mathbb{P}(\mathcal{C}_{\mathrm{S},1})$, $E[T_N(K) | C_{\mathrm{S},1}]$, and $E[T^2_N(K) | C_{\mathrm{S},1}]$ into (13) and (14), we obtain $E\left[X_n^{\mathrm{S}}\right]$ and $E\left[\left(X_n^{\mathrm{S}}\right)^2\right]$.

\subsection{First and Second Moments of $W$}
Recall that $W$ is the summation of $M-1$ consecutive inter-generation time $X_{n,j}^{\mathrm{F}}$, i.e., $W=\sum_{i=j}^{j+M-2}X_{n,i}^{\mathrm{F}}$.
As the probability that device $n$ successfully receives each status update is the same, $M$ is a geometric random variable.
As a result, the probability mass function (PMF) of $M$ is given by $\mathbb{P}(M = m)=(1-\mathbb{P}(\mathcal{C}_{\mathrm{S}}))^{m-1}\mathbb{P}(\mathcal{C}_{\mathrm{S}})$, $m \ge 1$.
Obviously, we have $\mathbb{E}[M]=\frac{1}{\mathbb{P}(\mathcal{C}_{\mathrm{S}})}$ and $\mathbb{E}[M^2]=\frac{2-\mathbb{P}(\mathcal{C}_{\mathrm{S}})}{(\mathbb{P}(\mathcal{C}_{\mathrm{S}}))^2}$.
As $M$ and $X_{n}^\mathrm{F}$ are independent, the first moment of $W$ can be calculated by 
 \begin{equation}
 \mathbb{E}[W]=(\mathbb{E}[M]-1)\mathbb{E}[X_{n}^\mathrm{F}].
 \label{ex20}
 \end{equation}

To derive the expression of $\mathbb{E}[W^2]$, we first calculate the variance of $W$ as follows
\begin{eqnarray}
 \mathrm{Var}[W] && \hspace{-6mm} =\mathrm{Var}\left[\mathbb{E}[W|M]\right]+\mathbb{E}\left[\mathrm{Var}[W|M]\right] \nonumber \\
 && \hspace{-6mm} =\left(\mathbb{E}\left[X_n^{\mathrm{F}}\right]\right)^2 \mathrm{Var}[M]+\mathrm{Var}[X_n^{\mathrm{F}}]\left(\mathbb{E}[M]-1\right), 
 \label{ex21}
\end{eqnarray}
where $\mathrm{Var}\left[ X_{n}^{\mathrm{F}} \right] = \mathrm{E}\left[ \left(X_n^{\mathrm{F}} \right)^2 \right] - \left( \mathbb{E}\left[ X_{n}^{\mathrm{F}} \right] \right)^2$. 
Based on \eqref{ex20} and \eqref{ex21}, we obtain $\mathbb{E}[W^2]=(\mathbb{E}[W])^2+\mathrm{Var}[W]$.

\subsection{First Moment of Service Time  $\hat{T}_{n,q}$}\label{SubSec_TNQ}
Recall that $\hat{T}_{n,q}$ is the service time of the $q$-th status update successfully delivered to device $n$.
Conditioning on the occurrence of Case $\mathcal{C}_{\mathrm{S}}$, the CDF of the service time is
\begin{equation}
	\begin{split}
		& F_{T | \mathcal{C}_{\mathrm{S}}}(t) =  \mathbb{P}(T_{n,j}<t | \mathcal{C}_{\mathrm{S}})  \\
		&= \frac {\sum_{h=1}^{K}\left( 1-  \mathcal{Z}_{h} -  \sum_{j=0}^{h-1} B_{h,j} \frac{e^{-\lambda_s(t-c)U_{h,j}}-V_{K,j}}{{U_{h,j}}}\right)}{{N \mathbb{P}(\mathcal{C}_{\mathrm{S}})}} ,
		 \label{ex23}
\end{split}
\end{equation}
where $U_{h,j}$, $V_{K,j}$, $B_{h,j}$, and $\mathcal{Z}_{h}$ are defined in Proposition 1. 
Based on (\ref{ex23}), the expectation of $\hat{T}_{n,q}$ can be calculated by 
\begin{equation}
	\begin{split}
		& \mathbb{E}[\hat{T}_{n,q}] =\int_{c}^{T_\mathrm{D}}t \,\,\, \mathrm{d}\, F_{T | \mathcal{C}_{\mathrm{S}}}(t) \\
		=& \sum_{h=1}^{K}\sum_{j=0}^{h-1} B_{h,j}  \frac{c\lambda_sU_{h,j}+1-e^{-\lambda_sU_{h,j}T_\mathrm{D}}(\lambda_sU_{h,j}T_\mathrm{D}+1)}{N\mathbb{P}(\mathcal{C}_{\mathrm{S}}) \lambda_s U_{h,j}^2}, 
	  \label{ex24}
\end{split}
\end{equation}
where $\mathbb{P}(\mathcal{C}_{\mathrm{S}}) = \frac{1}{N} \sum_{h=1}^{K}   \left(1 - \mathcal{Z}_{h}\right)$. 
%where $f_{T_{n,j} | \mathcal{C}_{\mathrm{S}}}(t)$ is the first derivative of $F_{T_{n,j} | \mathcal{C}_{\mathrm{S}}}(t)$.

% \begin{eqnarray}
%  \mathbb{E}[\hat{T}_{n,q}]&& \hspace{-6mm} =\int_{0}^{T_{\mathrm{D}}}t f_{T_{n,j} | \mathcal{C}_{\mathrm{S}}}(t) \mathrm
% {d} t  \nonumber \\
%  && \hspace{-6mm} =\frac{h}{N\mathbb{P}(\mathcal{C}_{\mathrm{S}})}\sum_{h=1}^{K} \binom{N}{h} \sum_{j=0}^{h-1}\binom{h-1}{j} (-1)^j \frac{1-e^{-\lambda_s\mathcal{R}T_{\mathrm{D}}}(\lambda_s\mathcal{R}T_{\mathrm{D}}+1)}{\lambda_s\mathcal{R}^2},
%  \label{ex24}
% \end{eqnarray}
%where $f_{T_{n,j} | \mathcal{C}_{\mathrm{S}}}(t)$ denotes the first derivative of $F_{T_{n,j} | \mathcal{C}_{\mathrm{S}}}(t)$.

 \subsection{Average AoI}
Finally, by substituting \eqref{ex30_1}, \eqref{ex30}, \eqref{ex20}, and \eqref{ex24} into \eqref{AveAoI}, we obtain the average AoI of the multicast transmission with hard deadlines.
% By substituting \eqref{ex30_1}, \eqref{ex30}, and \eqref{ex24} into \eqref{ex5}, we obtain $\mathbb{E}[A_{n,q}]$.
% In addition, by substituting \eqref{ex30} and \eqref{ex20} into \eqref{ex7} we obtain $\mathbb{E}[Y_{n,q}]$.
% Finally, we obtain the average AoI of the multicast transmission with hard deadlines, given in \eqref{ex3}.
It is worth pointing out that the results presented in this paper are general and can be easily extended to the scenarios for broadcast transmission with hard deadlines by replacing $K$ with $N$, for multicast transmission without hard deadlines by setting $T_{\mathrm{D}} = \infty$, and for unicast transmission with hard deadlines by setting $N = K =1$.

 \section{Performance Evaluation and Discussions}\label{sim}
In this section, we present the simulation and numerical results of the considered multicast transmission with hard deadlines. 
Unless specified otherwise, we set the total number of IoT devices $N=10$, the number of IoT devices required to successfully receive each status update $K=7$, and $c=0.1$. 

Fig. \ref{fig:n} shows the impact of hard deadline $T_{\mathrm{D}}$ on the average AoI for different values of average service rate $\lambda_s$.  
The simulation results match well with the theoretical results, which validates the performance analysis presented in Section \ref{k}. 
With the variation of deadline $T_{\mathrm{D}}$, the average AoI first decreases to a minimum value and then increases to a saturation value. 
Specifically, when $\lambda_s = 1/3$ and deadline $T_{\mathrm{D}}$ is small, the probability that each device can successfully receive a status update within a transmission interval (i.e., $\min \{T_N(K), T_{\mathrm{D}}\}$) is small. 
As such, it takes each IoT device many transmission intervals to successfully receive a status update. 
Note that the average AoI is proportional to the average number of transmission intervals required to successfully receive a status update and the average length of transmission intervals. 
Hence, the average AoI of the considered system is large when the deadline is small (e.g., 0.2). 
By increasing the value of deadline $T_{\mathrm{D}}$ to $0.9$, the average AoI declines quickly until reaching its minimum value. 
This is due to the fact that the probability of successful status update reception within each transmission interval increases. 
By further increasing the value of deadline $T_{\mathrm{D}}$, the average length  of transmission intervals increases and it starts to play a more important role in the AoI evolution than the average number of transmission intervals required to successfully a status update, leading to the increase of the average AoI. 
When deadline $T_{\mathrm{D}}$ is sufficiently large, the average AoI reaches a saturation value. 
The saturation value corresponds to the average AoI of multicast transmission without hard deadlines and is plotted with dashed lines in Fig. \ref{fig:n}. 
In addition, the average AoI decreases with the value of $\lambda_s$. This is because the average service time affects the average length of transmission intervals. 
Moreover, we can also observe that the value of deadline $T_{\mathrm{D}}$ that minimizes the average AoI becomes larger as the value of $\lambda_s$ decreases. 

%, and the hard deadline has a bigger impact on the average AoI when $T_s$ is larger. 

%Hence, the server needs to transmit next updates continuously till node $n$ receives the update. Those results in a high average AoI. 
%Specifically, when $T_s=3$, as the deadline increases to $0.8$ from $0.1$, the average AoI decline dramatically and reach the minimum point.
%This is because when deadlines closes to $c$, the probability that node $n$ receives updates successfully is very small and $\min\{T_N(K),T_{\mathrm{D}}\}$ is always equals to $T_{\mathrm{D}}$. At this moment, the amount of average AoI increases are the number of transmission frequency multiply the deadline. Therefore, the number of transmission frequency is large and dominates the average AoI. 
%Then, as deadlines grows, $\min\{T_N(K),T_{\mathrm{D}}\}$ not only equals to $T_{\mathrm{D}}$, may equal to $T_N(K)$. The transmission frequency cannot dominate the average AoI gradually. 
%At this moment, the average AoI impacted by transmission frequency and $T_N(K)$ simultaneously.
%Around $T_{\mathrm{D}}=0.8$, both transmission frequency and $T_N(K)$ are small, thus average AoI reach minimum value. 
%When $T_{\mathrm{D}} > 0.8$, the transmission frequency increases slightly. 
%However, $T_N(K)$ increases significantly with the deadline grows. Therefore, the average AoI increases gradually with the deadline.
%When $T_\mathrm{D}>6$, both the transmission frequency and $T_N(K)$ reach the maximum value. Thus,the average AoI cannot increase anymore. 

 \begin{figure}
   \centering
   \includegraphics[scale=0.5]{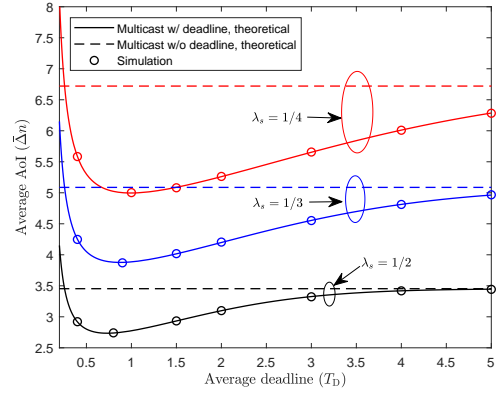}
   \caption{Average AoI versus deadline $T_{\mathrm{D}}$ for different values of $\lambda_s$ when $K=7$, $N = 10$, and $c=0.1$.}
   \label{fig:n}
    \vspace{-5mm}
 \end{figure}

%As noted in \eqref{ex12} that $\mathbb{P}(\mathcal{C}_{\mathrm{F},2} | \mathcal{C}_{\mathrm{F}} )$ close to zero, $X_{n}^\mathrm{F}$ tends to $T_\mathrm{D}$ when the deadline is small. Hence, average AoI is highly sensitive to the deadline. When the deadline grows, obviously $X_{n}^\mathrm{F}$ is affected by the $K$ in \eqref{ex12}. Similarly, when the deadline is sufficiently large, $\mathbb{P}(\mathcal{C}_{\mathrm{F},1} | \mathcal{C}_{\mathrm{F}} )$ close to zero, $X_{n}^\mathrm{F}$ insensitive to $T_\mathrm{D}$. Average AoI tends to a fixed value.

 \begin{figure}
   \centering
   \includegraphics[scale=0.5]{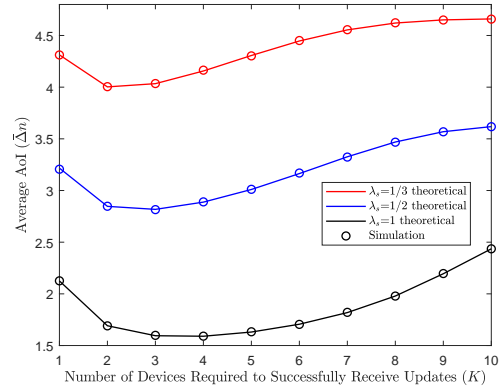}
   \caption{Average AoI versus number of IoT devices required to successfully receive each status update for different values of $\lambda_s$ when $T_{\mathrm{D}}=3$, $N = 10$, and $c=0.1$. }
   \label{fig:s}
   \vspace{-5mm}
 \end{figure}

Fig. \ref{fig:s} illustrates the impact of $K$ on the average AoI of the considered system for different values of $\lambda_s$ when $T_{\mathrm{D}} = 3$. 
When $K$ is small (e.g., $K = 1$), the probability that a specific device is one of the first $K$ devices that successfully receive the status update is low, and hence the average AoI is relatively large. 
When $\lambda_s = 1/2$, by increasing the value of $K$ to 3, the probability of successful status update reception increases, which reduces the number of transmission intervals that are required to successfully receive a status update and in turn reduces the average AoI. 
By further increasing the value of $K$, the average length of transmission intervals increases as more devices are required to successfully receive each status update. 
As the average length of transmission intervals increasingly dominates the AoI evolution when $K \ge 4$, the average AoI increases. 
Therefore, with the variation of $K$, there exists a value of $K$ that balances the tradeoff between these two effects and minimizes the average AoI. 
%{\color{red}It is useful when we design multicast networks. In particular, when the $K$ in some applications is less than the optimal $K$, we can increase $K$ to the optimal value to ensure more devices can receive status updates successfully and makes the average AoI is lower.  }
%Moreover, the optimal value of $K$ increases as the value of $\lambda_s$ increases. 
Similarly, we can observe that the average AoI increases as the value of $\lambda_s$ decreases.

 \section{Conclusions}\label{con}
In this paper, we studied the average AoI of multicast transmission with hard deadlines in IoT systems.
We characterized the instantaneous AoI evolution and derived the first and second moments of the inter-generation time of two consecutive status updates for both successful and unsuccessful reception cases.
We also derived the first and second moments of the time duration within which all the status updates transmitted by the server are not successfully received by a specific device.
Based on these derived expressions, we obtained the closed-form expression of the average AoI.
Simulations validated the theoretical analysis and illustrated the impact of system parameters on the average AoI.  
Results showed that the hard deadline has a significant impact on the average AoI of multicast transmission and there exists an optimal value of the hard deadline that minimizes the average AoI. 

\section*{Appendix}
\subsection{Proof of Proposition 1}
When Case $\mathcal{C}_{\mathrm{F},1}$ occurs, we have $T_{N}(K) < T_{\mathrm{D}}$ and $T_{n,j} > \min\{T_{\mathrm{D}}, T_{N}(K) \}$, which can be simplified as $T_{N}(K) < \min \{ T_{\mathrm{D}}, T_{n,j} \}$.
As a result, the CDF of the time that $K$ IoT devices successfully receive a status update conditioning on the occurrence of Case $\mathcal{C}_{\mathrm{F},1}$ can be expressed as
\begin{equation}
	\begin{split}
		& F_{T_N(K)|\mathcal{C}_{\mathrm{F},1}}(t) = \mathbb{P} \left( T_N(K)<t|\mathcal{C}_{\mathrm{F},1}\right)\\
		= & \frac{\mathbb{P}\left( T_N(K)<t,   T_{N}(K) < \min \{ T_{\mathrm{D}}, T_{n,j} \}\right)} {\mathbb{P}( T_{N}(K) < \min \{ T_{\mathrm{D}}, T_{n,j} \})}.
	\label{ex8}
	\end{split}
\end{equation}
The numerator of \eqref{ex8} can be calculated as
\begin{eqnarray}
\hspace{-5mm}	 && \hspace{-6mm} \mathbb{P}\left( T_N(K)<t,  T_{N}(K) < \min \{ T_{\mathrm{D}}, T_{n,j} \}\right) \nonumber \\
\hspace{-5mm}	= && \hspace{-6mm} \frac{N-K}{N} \left(1- \mathcal{Z}_{K}- \sum_{j=0}^{K-1}B_{K,j}\frac{ e^{-\lambda_s(t-c)U_{K,j}}-V_{K,j}}{{U_{K,j}}}\right),
	   \label{ex9}
\end{eqnarray}
% \begin{eqnarray}
% \mathbb{P}\left(T_N(K)<t, T_{N}(K) < \min \{ T_{\mathrm{D}}, T_{n,j} \}\right) \nonumber\\
% && \hspace{-6mm} =\frac{N-K}{N}\mathbb{P}(T_N(K)<t,T_N(K)<T_{\mathrm{D}}) \nonumber \\
%   && \hspace{-6mm} =\frac{N-K}{N}\left(1-\mathcal{Z}_{K,j}\right)- K\binom{N-1}{K} \sum_{j=0}^{K-1}\binom{K-1}{j} (-1)^j \frac{ \mathcal{V}(t)}{\mathcal{R}},
%   \label{ex9}
% \end{eqnarray}
where $V_{K,j}$, $B_{K,j}$, $U_{K,j}$, and $\mathcal{Z}_{K}$ are defined in Proposition 1.
% , $\mathbb{P}\left(T_{n,j}> T_N(K) \right)=\frac{N-K}{N}$ as $\{ T_{n,j}, n \in \mathcal{N} \}$ are identically distributed.

On the other hand, the denominator of \eqref{ex8} is given by
\begin{equation}
 \begin{split}
	 \mathbb{P}\left( T_{N}(K) < \min \{ T_{\mathrm{D}}, T_{n,j} \} \right)
	 = \frac{N-K}{N}\left(1- \mathcal{Z}_{K}\right).
	 \label{ex10}
 \end{split}
\end{equation}
% \begin{eqnarray}
%  \mathbb{P}\left(T_{N}(K) < \min \{ T_{\mathrm{D}}, T_{n,j} \} \right)&& \hspace{-6mm}=\frac{N-K}{N} \mathbb{P}\left(T_N(K)\le T_{\mathrm{D}} \right) = \frac{N-K}{N}\left(1-\mathcal{Z}_{K,j}\right).
%  \label{ex10}
% \end{eqnarray}

By substituting (\ref{ex9}) and (\ref{ex10}) into (\ref{ex8}), we obtain the conditional CDF of $T_N(K)$, i.e., $F_{T_N(K)|\mathcal{C}_{\mathrm{F},1}}(t)$, as follows
\begin{equation}
 \begin{split}
\hspace{-3mm} F_{T_N(K)|\mathcal{C}_{\mathrm{F},1}}(t) \!= 1 - \frac{1}{1-\mathcal{Z}_{K}} \sum_{j=0}^{K-1} B_{K,j}\frac{ e^{-\lambda_s(t-c)U_{K,j}}-V_{K,j}}{U_{K,j}}.
 \end{split}
 \label{ex100}
\end{equation}

As a result, the conditional first and second moments of $T_N(K)$,
 denoted as $\mathbb{E} \left[ T_N(K) \left| \mathcal{C}_{\mathrm{F},1} \right. \right]$ and $\mathbb{E} \left[ T_N(K)^2 \left| \mathcal{C}_{\mathrm{F},1} \right. \right]$,
can be written as
\begin{equation}
	\begin{split}
	&\mathbb{E}\left[T_N(K) \left| \mathcal{C}_{\mathrm{F},1} \right. \right]  =\int_{c}^{T_{\mathrm{D}}} t f_{T_N(K)|\mathcal{C}_{\mathrm{F},1}}(t) \mathrm{d}t \\
 &= \frac{\sum_{j=0}^{K-1} \frac{B_{K,j} }{\lambda_s U_{K,j}^2}\Big[1+c\lambda_s U_{K,j}  - (1 + T_{\mathrm{D}}\lambda_s U_{K,j}  )V_{K,j}\Big]}{{1- \mathcal{Z}_{K}}}, 
 	\label{ex11}
	\end{split}
\end{equation}
\begin{equation}
 \begin{split}
   \mathbb{E}\left[T_N^2(K) \left| \mathcal{C}_{\mathrm{F},1} \right. \right]   & = \int_{c}^{T_{\mathrm{D}}} t^2 f_{T_N(K)|\mathcal{C}_{\mathrm{F},1}}(t) \mathrm{d}t \\
   &=\frac{1}{1- \mathcal{Z}_{K}} \sum_{j=0}^{K-1} \frac{B_{K,j}}{\lambda_s^2 U_{K,j}^3} \Big[ (1 + c \lambda_s U_{K,j})^2 \\
   & \hspace{3mm} - \left((1+T_{\mathrm{D}}\lambda_s U_{K,j}  )^2 +1\right)V_{K,j} \Big],
  \end{split}
\label{ex51}
\end{equation}
where $\mathcal{Z}_{K}$, $B_{K,j}$, $U_{K,j}$, and $V_{K,j}$ are given in Proposition 1, and $f_{T_N(K)|\mathcal{C}_{\mathrm{F},1}}(t)$ is the first derivative of $F_{T_N(K)|\mathcal{C}_{\mathrm{F},1}}(t)$ and denotes the conditional PDF of $T_N(K)$.

\subsection{Proof of Proposition 2}
The occurrence probability of Case $\mathcal{C}_{\mathrm{F},2}$ is given by 
\begin{eqnarray} \label{Eq_CF2}
\mathbb{P}(\mathcal{C}_{\mathrm{F},2}) && \hspace{-6mm} = \frac{ \mathbb{P}\left(T_{\mathrm{D}} < \min \{ T_N(K), T_{n,j} \} \right) } { \mathbb{P}\left(T_{n,j}>\min \{ T_{\mathrm{D}}, T_N(K) \} \right) },
\label{ex12}
\end{eqnarray}
where the denominator $\mathbb{P}\left(T_{n,j}>\min \{ T_{\mathrm{D}}, T_N(K) \} \right) = \mathbb{P}\left(T_{n,j}>T_{\mathrm{D}} \right)+\mathbb{P} \left(T_{n,j}> T_N(K) \right)-\mathbb{P}\left(T_{n,j}>T_{\mathrm{D}},T_{n,j}> T_N(K) \right)$. 
By definition, we have $\mathbb{P}\left(T_{n,j}>T_{\mathrm{D}} \right) = \mathrm{e}^{-\lambda_s (T_{\mathrm{D}}-c)}$ and $\mathbb{P}\left(T_{n,j}> T_N(K) \right)=\frac{N-K}{N}$. 
Besides, the probability that the $T_{n,j}$ is greater than both $T_{\mathrm{D}}$ and $T_N(K)$ is given by
\begin{equation}
 \begin{split}
	 &\mathbb{P}\left(T_{n,j}>T_{\mathrm{D}},T_{n,j}>T_N(K) \right) \\
	 = &\sum_{h=K+1}^{N} \mathbb{P} \left(T_{n,j}>T_{\mathrm{D}},T_{n,j}=T_N(h) \right)= \frac{1}{N}\sum_{h=K+1}^{N}\mathcal{Z}_{h},
	\label{ex13}
 \end{split}
\end{equation}
% \begin{eqnarray}
%  \mathbb{P}\left(T_{n,j}>T_{\mathrm{D}},T_{n,j}>T_N(K) \right) && \hspace{-6mm} = \sum_{h=K+1}^{N} \mathbb{P} \left(T_{n,j}>T_{\mathrm{D}},T_{n,j}=T_N(h) \right) =\frac{1}{N}\sum_{h=K+1}^{N}\mathcal{Z}_{h,j},
% \label{ex13}
% \end{eqnarray}
where $\mathcal{Z}_{h}$ is defined in Proposition 1. 
On the other hand, the numerator of (\ref{ex12}) can be calculated by
\begin{equation}
 \begin{split}
	 &\mathbb{P}\left(T_{\mathrm{D}} <  \min \{T_N(K),T_{n,j} \} \right) = \frac{N-K}{N}\mathcal{Z}_{K}+\sum_{h=1}^{K}  \mathcal{Z}_{h}.
	 \label{ex15}
 \end{split}
\end{equation}
% \begin{eqnarray}
% \mathbb{P}\left(T_{\mathrm{D}} <  \min \{T_N(K),T_{n,j} \} \right) && \hspace{-6mm} =\mathbb{P}\left(T_{\mathrm{D}} < T_N(K), T_N(K) < T_{n,j} \right)+\mathbb{P} \left(T_{\mathrm{D}}<T_{n,j}, T_{n,j}\le T_N(K) \right) \nonumber \\
% && \hspace{-6 mm} = \frac{N-K}{N}\mathcal{Z}_{K,j}+\sum_{h=1}^{K} \mathcal{Z}_{h,j},
% \label{ex15}
% \end{eqnarray}
% where $\mathbb{P}\left(T_{\mathrm{D}}<T_N(h) \right)$ is given in (\ref{ex14}).
% Hence, by combining (\ref{ex13}), (\ref{ex14}), and (\ref{ex15}), we obtain the conditional probability $\mathbb{P} \left( \mathcal{C}_{\mathrm{F},2} | \mathcal{C}_{\mathrm{F}} \right)$ in (\ref{ex12}).
By substituting \eqref{ex13} and \eqref{ex15} into \eqref{ex12}, we obtain $\mathbb{P} \left( \mathcal{C}_{\mathrm{F},2}\right)$. 

\bibliographystyle{IEEEtran}

\bibliography{refs}

\end{document}